\documentclass[]{article}
\usepackage{amsmath,amsthm,amssymb,mathtools,amsfonts}
\usepackage{mathrsfs}
\usepackage[numbers]{natbib}
\usepackage{hyperref,cleveref}
\usepackage{authblk}
\usepackage{color}
\usepackage{xcolor}
\usepackage[a4paper]{geometry}

\usepackage{placeins}
\newtheorem{theorem}{Theorem}

\newtheorem{lemma}{Lemma}

\usepackage{algorithm}
\usepackage{algpseudocode}

\newcommand{\T}{\top}
\newcommand{\Lb}{\mathcal{L}}
\newcommand{\Ub}{\mathcal{U}}
\newcommand{\select}{\mathsf{SELECT}}
\newcommand{\val}{\operatorname{val}}

\newcommand{\opt}{\operatorname{\mathsf{SELECT}}}
\usepackage{tikz}
\DeclareMathOperator{\Exp}{\mathbb{E}}
\definecolor{soph}{RGB}{130,17,166}

\newcommand{\tmedian}{\mathsf{2median}}

\usepackage[T1]{fontenc}
\usepackage[utf8]{inputenc}
\usepackage[english]{babel}
\usepackage{float}
\algrenewcommand{\alglinenumber}[1]{\color{gray}\footnotesize#1:}

\title{A nearly optimal randomized algorithm for explorable heap selection}

\author[1]{Sander Borst\thanks{This project has received funding from the European Research Council (ERC) under the European Union's Horizon 2020 research and innovation programme (grant agreement QIP--805241)}}
\author[1]{Daniel Dadush$^*$}
\author[2]{Sophie Huiberts}
\author[1]{Danish Kashaev}

\affil[1]{Centrum Wiskunde \& Informatica (CWI), Amsterdam}
\affil[2]{Columbia University, New York}
\affil[ ]{\texttt {\{sander.borst,dadush,danish.kashaev\}@cwi.nl}, \texttt{sophie@huiberts.me}}
\usepackage{eqparbox}

\begin{document}
\maketitle
\begin{abstract}
Explorable heap selection is the problem of selecting the $n$th smallest value in a binary heap. The key values can only be accessed by traversing through the underlying infinite binary tree, and the complexity of the algorithm is measured by the total distance traveled in the tree (each edge has unit cost). This problem was originally proposed as a model to study search strategies for the branch-and-bound algorithm with storage restrictions by Karp, Saks and Widgerson (FOCS '86), who gave deterministic and randomized $n\cdot \exp(O(\sqrt{\log{n}}))$ time algorithms using $O(\log(n)^{2.5})$ and $O(\sqrt{\log n})$ space respectively.
We present a new randomized algorithm with running time $O(n\log(n)^3)$ against an oblivious adversary using $O(\log n)$ space, substantially improving the previous best randomized running time at the expense of slightly increased space usage. We also show an $\Omega(\log(n)n/\log(\log(n)))$ lower bound for any algorithm that solves the problem in the same amount of space, indicating that our algorithm is nearly optimal.
\end{abstract}

\section{Introduction}

Many important problems in theoretical computer science are fundamentally
search problems.  The objective of these problems is to find a certain
solution from the search space.  In this paper we analyze a search problem
that we call \emph{explorable heap selection}. The problem is related to the
famous branch-and-bound algorithm and was originally proposed by Karp, Saks and Widgerson~\cite{KSW86} to model node
selection for branch-and-bound with low space-complexity.
Furthermore, as we will explain later, the problem remains practically
relevant to branch-and-bound even in the full space setting.

The explorable heap selection problem\footnote{~\cite{KSW86} did not give the
problem a name, so we have attempted to give a descriptive one here.} is an
online graph exploration problem for an agent on a rooted (possibly infinite) binary tree. The nodes of the tree are labeled by distinct real numbers (the key values) that increase along every
path starting from the root. The tree can thus be thought of as a min-heap. Starting at the root, the agent's objective is to select the
$n^{\text{th}}$ smallest value in the tree while minimizing the distance traveled, where each edge of the tree has unit travel cost. The key value of a
node is only revealed when the agent visits it, and thus the problem has an
online nature. When the agent learns the key value of a node, it still does not know the rank of this value.

The selection problem for ordinary heaps, which allow for random access
(i.e., jumping to arbitrary nodes in the tree for ``free''), has also been
studied. In this model, it was shown by~\cite{frederickson_optimal_1993} that
selecting the $n^{\text{th}}$ minimum can be achieved deterministically in
$O(n)$ time using $O(n)$ workspace. We note that in both models, $\Omega(n)$
is a natural lower bound. This is because verifying that a value $\Lb$ is the
$n^{\text{th}}$ minimum requires $\Theta(n)$ time -- one must at least inspect the $n$
nodes with value at most $\Lb$ -- which can be done via straightforward
depth-first search.

A simple selection strategy is to use the best-first
rule\footnote{Frederickson's algorithm~\cite{frederickson_optimal_1993} is in
fact a highly optimized variant of this rule}, which repeatedly explores the
unexplored node whose parent has the smallest key value.  While this rule is
optimal in terms of the number of nodes that it explores, namely $\Theta(n)$,
the distance traveled by the agent can be far from optimal. In the
worst-case, an agent using this rule will need to travel a distance of
$\Theta(n^2)$ to find the $n^{\text{th}}$ smallest value. A simple bad example
for this rule is to consider a rooted tree consisting of two paths (which one
can extend to a binary tree), where the two paths are consecutively labeled
by all positive even and odd integers respectively. Moreover, the
space complexity becomes $\Omega(n)$ in general when using the best-first rule, because
essentially all the explored nodes might need to be kept in memory.
We note that irrespective of computational considerations on the agent, either in
terms of working memory or running time restrictions, minimizing the total
travel distance in explorable heap selection remains a challenging online
problem.

Improving on the best-first strategy, Karp, Saks and Wigderson~\cite{KSW86}
gave a randomized algorithm with expected cost $n\cdot
\exp(O(\sqrt{\log(n)}))$ using
$O(\sqrt{\log(n)})$ working space. They also showed how to make
the algorithm deterministic using $O(\log(n)^{2.5})$ space. In this work, our
main contribution is an improved randomized algorithm with expected cost
$O(n\log(n)^3)$ using $O(\log(n))$ space. Given the $\Omega(n)$
lower bound, our travel cost is optimal up to logarithmic factors. Furthermore we show that any algorithm for explorable
heap selection that uses only $s$ units of memory, must take at least $n\cdot \log_s (n)$ time in expectation. An interesting open problem is the question whether a superlinear lower bound also holds without any restriction on the memory usage.

To clarify the memory model, it is assumed that any key value and $O(\log n)$
bit integer can be stored using $O(1)$ space. We also assume that maintaining the current position in the tree does not take up memory. Furthermore, we assume
that key value comparisons and moving across an edge of the tree
require $O(1)$ time. Under these assumptions, the running times of the above
algorithms happen to be proportional to their travel cost. Throughout the
paper, we will thus use travel cost and running time interchangeably.

\paragraph{Motivation}
The motivation to look at this problem comes from the branch-and-bound algorithm. This is a well-known algorithm that can be used for solving many types of problems.
In particular, it is often used to solve integer linear programs (IPs), which are of the form $\arg\min \{ c^\T x : x\in \mathbb{Z}^n,  Ax\leq b\}$.
In that setting, branch-and-bound works by first solving the linear programming (LP) relaxation, which does not have integrality constraints. The value of the solution to the relaxation forms a lower bound on the objective value of the original problem. Moreover, if this solution only has integral components, it is also optimal for the original problem.
Otherwise, the algorithm chooses a component $x_i$ for which the solution value $\hat{x}_i$ is not integral. It then creates two new subproblems, by either adding the constraint $x_i\leq \lfloor \hat{x}_i\rfloor$ or $x_i\geq \lceil \hat{x}_i\rceil$. This operation is called \emph{branching}. The tree of subproblems, in which the children of a problem are created by the branching operation, is called the branch-and-bound tree. Because a subproblem contains more constraints than its parent, its objective value is greater or equal to the one of its parent. The algorithm can also be used to solve mixed-integer linear programs (MIPs), where some of the variables are allowed to be continuous.

At the core, the algorithm consists of two important components: the branching rule and the node selection rule.
The branching rule determines how to split up a problem into subproblems, by choosing a variable to branch on. Substantial research has been done on branching rules, see, e.g., \cite{linderoth_computational_1999-1,achterberg_branching_2005,lodi_learning_2017-1,balcan_learning_2018}.

The node selection rule decides which subproblem to solve next.
Not much theoretical research has been done on the choice of the node selection rule.
Traditionally, the best-first strategy is thought to be optimal from a theoretical perspective because this rule minimizes the number of nodes that need to be visited.
However, a disadvantage of this rule is that searches using it might use space proportional to the number of explored nodes, because all of them need to be kept in memory. In contrast to this, a simple strategy like depth-first search only needs to store the current solution. Unfortunately, performing a depth-first search can lead to an arbitrarily bad running time.
This was the original motivation for introducing the explorable heap selection problem \cite{KSW86}. By guessing the number $N$ of branch-and-bound nodes whose LP values are at most that of the optimal IP solution (which can be done via successive doubling), a search strategy for this problem can be directly interpreted as a node selection rule.
The algorithm that they introduced can therefore be used to implement branch-and-bound efficiently in only $O\left(\sqrt{\log(N)}\right)$ space.

Nowadays, computers have a lot of memory available. This usually makes it feasible to store all explored nodes of the branch-and-bound tree in memory.
However, many MIP-solvers still make use of a hybrid method that consists of both depth-first and best-first searches.
This is not only done because depth-first search uses less memory, but also because it is often faster.
Experimental studies have confirmed that the depth-first strategy is in many cases faster than best-first one \cite{clausen_best_nodate}.
This seems contradictory, because the running time of best-first search is often thought to be theoretically optimal.

In part, this contradiction can be explained by the fact that actual IP-solvers often employ complementary techniques and heuristics on top of branch-and-bound, which might benefit from depth-first searches.
Additionally, a best-first search can hop between different parts of the tree, while a depth first search subsequently explores nodes that are very close to each other.
In the latter case, the LP-solver can start from a very similar state, which is known as warm starting. This is faster for a variety of technical reasons \cite{achterberg2009}.
For example, this can be the case when the LP-solver makes use of the LU-factorization of the optimal basis matrix \cite{morrison_branch-and-bound_2016-1}.
Through the use of dynamic algorithms, computing this can be done faster if a factorization for a similar LP-basis is known \cite{suhl_fast_1993}. 
Because of its large size, MIP-solvers will often not store the LU-factorization for all nodes in the tree.
This makes it beneficial to move between similar nodes in the branch-and-bound tree.
Furthermore, moving from one part of the tree to another means that the solver needs to undo and redo many bound changes, which also takes up time.
Hence, the amount of distance traveled between nodes in the tree is a metric that influences the running time. This can also be observed when running the academic MIP-solver SCIP \cite{gleixner_personal}.

The explorable heap selection problem captures these benefits of locality by measuring the running
time in terms of the amount of travel through the tree. Therefore, we argue that this
problem is still relevant for the choice of a node selection rule, even if all nodes
can be stored in memory.

\paragraph{Related work}
The explorable heap selection problem was first introduced in \cite{KSW86}.
Their result was later applied to prove an upper bound on the parallel running time of branch-and-bound \cite{pietracaprina_space-ecient_nodate}.

When random access to the heap is provided at constant cost,
selecting the $n$'th value in the heap can be done by a deterministic algorithm
in $O(n)$ time by using an additional $O(n)$ memory for auxilliary data structures
\cite{frederickson_optimal_1993}.

The explorable heap selection problem can be thought of as a \emph{search game}
\cite{theoryofsearch} and bears some similarity to the \emph{cow path problem}.
In the cow path problem, an agent explores an unweighted unlabeled graph in
search of a target node.
The location of the target node is unknown, but when the agent visits a node
they are told whether or not that node is the target. The performance of an
algorithm is judged by the ratio of the number of visited nodes to the distance of the
target from the agent's starting point.
In both the cow path problem and the explorable heap selection problem,
the cost of backtracking and retracing paths is an important consideration.
The cow path problem on infinite $b$-ary trees was studied in \cite{goos_near_1995}
under the assumption that when present at a node the agent can obtain an estimate
on that node's distance to the target.

Other explorable graph problems exist without a target, where typically the graph itself is unknown
at the outset. There is an extensive literature on exploration both in graphs and
in the plane \cite{Berman1998,kamphans}. In some of the used models the objective is to minimize the distance traveled \cite{banerjee_graph_2023,baligacs_exploration_2023,megow_online_2012,kalyanasundaram_constructing_1994}.
Other models are about minimizing the amount of used memory \cite{diks_tree_2004}.
What distinguises the explorable heap selection problem from these problems is the information that the graph is a heap and that the ordinal of the target is known.
This can allow an algorithm to rule out certain locations for the target. Because of this additional information, the techniques used here do not seem to be applicable to these other problems.

\paragraph{Outline}
In \Cref{sec:problem} we formally introduce the explorable heap selection problem and any notation we will use. In \Cref{sec:algorithm} we introduce a new algorithm for solving this problem and provide a running time analysis. In \Cref{sec:lowerbound} we give a lower bound on the complexity of solving explorable heap selection using a limited amount of memory.

\section{The explorable heap selection problem}
\label{sec:problem}
We introduce in this section the formal model for the explorable heap selection
problem. The input to the algorithm is an infinite binary tree $T = (V,E)$,
where each node $v\in V$ has an associated real value, denoted by
$\text{val}(v) \in \mathbb{R}$. We assume that all the values are distinct.
Moreover, for each node in the tree, the values of its children are larger than
its own value. Hence, for every $v_1, v_2 \in V$ such that $v_1$ is an ancestor
of $v_2$, we have that $\text{val}(v_2) > \text{val}(v_1)$. The binary tree $T$
is thus a heap.

The algorithmic problem we are interested in is finding the $n^{\text{th}}$
smallest value in this tree. This may be seen as an online graph exploration
problem where an agent can move in the tree and learns the value of a node each
time they explore it. At each time step, the agent resides at a vertex $v \in
V$ and may decide to move to either the left child, the right child or the
parent of $v$ (if it exists, i.e. if $v$ is not the root of the tree). Each
traversal of an edge costs one unit of time, and the complexity of an algorithm
for this problem is thus measured by the total traveled distance in the binary
tree. The algorithm is also allowed to store values in memory.

We now introduce a few notations used throughout the paper.
\begin{itemize}
    \item For a node $v\in V$, also per abuse of notation written $v \in T$, we denote by $T^{(v)}$ the subtree of $T$ rooted at $v$.
    \item For a tree $T$ and a value $\Lb \in \mathbb{R}$, we define the subtree $T_{\Lb}:= \{v \in T \mid \text{val}(v) \leq \Lb\}.$
    \item We denote the $n^{\text{th}}$ smallest value in $T$ by $\opt^T(n)$. This is the quantity that we are interested in finding algorithmically.
    \item We say that a value $\mathcal{V} \in \mathbb{R}$ is \emph{good} for a tree $T$ if $\mathcal{V} \leq \opt^T(n)$ and \emph{bad} otherwise. Similarly, we call a node $v \in T$ \emph{good} if $\text{val}(v) \leq \opt^T(n)$ and \emph{bad} otherwise.
    \item We will use $[k]$ to refer to the set $\{1,\ldots, k\}$. 
    \item When we write $\log(n)$, we assume the base of the logarithm to be $2$.
\end{itemize}

For a given value $\mathcal{V} \in \mathbb{R}$, it is easy to check whether it is good in $O(n)$ time: start a depth first search at the root of the tree, turning back each time a value strictly greater than $\mathcal{V}$ is encountered. In the meantime, count the number of values below $\mathcal{V}$ found so far and stop the search if more than $n$ values are found. If the number of values below $\mathcal{V}$ found at the end of the procedure is at most $n$, then $\mathcal{V}$ is a good value. This procedure is described in more detail later in the \Call{DFS}{} subroutine.

We will often instruct the agent to move to an already discovered good vertex $v \in V$. The way this is done algorithmically is by saving $\text{val}(v)$ in memory and starting a depth first search at the root, turning back every time a value strictly bigger than $\text{val}(v)$ is encountered until finally finding $\text{val}(v)$. This takes at most $O(n)$ time, since we assume $v$ to be a good node. If we instruct the agent to go back to the root from a certain vertex $v \in V$, this is simply done by traveling back in the tree, choosing to go to the parent of the current node at each step.

In later sections, we will often say that a subroutine takes a subtree $T^{(v)}$
as input. This implicitly means that we in fact pass it $\text{val}(v)$ as
input, make the agent travel to $v \in T$ using the previously described
procedure, call the subroutine from that position in the tree, and travel back
to the original position at the end of the execution. Because the subroutine
knows the value $\text{val}(v)$ of the root of $T^{(v)}$, it can ensure it never leaves the subtree
$T^{(v)}$, thus making it possible to recurse on a subtree as if it were a
rooted tree by itself. We write the subtree $T^{(v)}$ as part of the input for
simplicity of presentation.

We will sometimes want to pick a value uniformly at random from a set of values
$\{\mathcal{V}_1, \dots, \mathcal{V}_k\}$ of unknown size that arrives in a streaming fashion,
for instance when we traverse a part of the tree $T$ by doing a depth first
search. That is, we see the value $\mathcal{V}_i$ at the $i^{\text{th}}$ time
step, but do not longer have access to it in memory once we move on to
$\mathcal{V}_{i+1}$. This can be done by generating random values $\{X_1,
\dots, X_k\}$ where, at the $i^{\text{th}}$ time step, $X_i = \mathcal{V}_i$
with probability $1/i$, and $X_i = X_{i-1}$ otherwise. It is easy to check that
$X_k$ is a uniformly distributed sample from $\{\mathcal{V}_1, \dots,
\mathcal{V}_k\}$.

\section{A new algorithm}
\label{sec:algorithm}
The authors of \cite{KSW86} presented a deterministic algorithm  that solves
the explorable heap selection problem in $n\cdot \exp(O(\sqrt{\log(n)}))$ time
and $O(n\sqrt{\log(n)})$ space. By replacing the binary search that is used in
the algorithm by a randomized variant, they are able to decrease the space
requirements. This way, they obtain a randomized algorithm with expected
running time $n\cdot \exp(O(\sqrt{\log(n)}))$ and space complexity
$O(\sqrt{\log(n)})$. Alternatively, the binary search can be implemented using
a deterministic routine by \cite{MP80} to achieve the same running time with
$O(\log(n)^{2.5})$ space.

We present a randomized algorithm with a running time $O(n\log(n)^3)$ and space
complexity $O(\log(n))$. Unlike the algorithms mentioned before, our algorithm
fundamentally relies on randomness to bound its running time. This bound only
holds when the algorithm is run on a tree with labels that are fixed before the
execution of the algorithm. That is, the tree must be generated by an adversary that is oblivious to the choices made by the algorithm. This is a stronger assumption than is needed for
the algorithm that is given in \cite{KSW86}, which also works against adaptive adversaries. An adaptive adversary is able to defer the decision of
the node label to the time that the node is explored.  Note that this distinction does not really matter for the application of the algorithm as a node selection rule in branch-and-bound, since there the node labels are fixed because they are derived from the integer program and branching rule.

\begin{theorem}
There exists a randomized algorithm that solves the explorable heap selection
problem, with expected running time $O(n\log(n)^3)$ and $O(\log(n))$
space.
\end{theorem}

As mentioned above, checking whether a value $v$ is good can be done in $O(n)$ time by doing a
depth-first search with cutoff value $\val(v)$ that returns when more than $n$
good nodes are found. For a set of $k$ values, we can determine which of them
are good in $O(\log(k)n)$ time by performing a binary search.

The explorable heap selection problem can be seen as the problem of finding all
$n$ good nodes. Both our method and that of \cite{KSW86} function by first
identifying a subtree consisting of only good nodes. The children of the leaves
of this subtree are called ``roots'' and the subtree is extended by finding
a number of new good nodes under these roots in multiple rounds. Importantly, the term `good node' is always used with respect to the current call to \Call{Extend}{}.
So, a node might be good in one recursive call, but not good in another.

In \cite{KSW86} this is done by running $O(c^{\sqrt{2\log(n)}}
)$ different rounds, for some constant $c>1$. In each round, the algorithm
finds $n/c^{\sqrt{2\log(n)}}$ new good nodes. These nodes are found by
recursively exploring each active root and using binary search on the observed values
to discover which of these values are good. Which active roots are recursively explored
further depends on which values are good. The recursion in the algorithm
is at most $O(\sqrt{\log(n)})$ levels deep, which is where the space complexity
bound comes from.

In our algorithm, we take a different approach. We will call our algorithm consecutively with $n=1,2,4,8,\dots$. Hence, for a call to the algorithm, we can assume that we have already found at least $n/2$ good nodes. These nodes form a subtree of the original tree $T$.
In each round, our algorithm chooses a random root under this subtree
and finds every good node under it. It does so by doing recursive subcalls to the main algorithm on this root with values $n=1,2,4,8,\ldots$. As soon as the recursively obtained node is a bad node, the algorithm stops searching the subtree of this root, since it is guaranteed that all the good nodes there have been found. The largest good value that is found can then be used to find additional good nodes under the other roots without recursive calls, through a simple depth-first search.  Assuming that
the node values were fixed in advance, we expect this largest good value
to be greater than half of the other roots' largest good values.
Similarly, we expect its smallest bad value to be smaller than half of the other roots' smallest bad values.
By this principle, a sizeable fraction of the roots can, in expectation, be ruled out from getting a recursive call.
Each round a new random root is selected until all good nodes have been found.

This algorithm allows us to effectively perform binary search on the list of
roots, ordered by the largest good value contained in each of their subtrees
in $O(\log n)$ rounds, and the same list ordered by the smallest bad values (\Cref{lem:exp_iterations}).
Bounding the expected number of good nodes found using
recursive subcalls requires a subtle induction on two parameters (\Cref{lem:exp_sum}):
both $n$ and the number of good nodes that have been identified so far.

\subsection{Subroutines}
We first describe three subroutines that will be used in our main algorithm.

\paragraph{The procedure DFS} The procedure \Call{DFS}{} is a variant of depth
first search. The input to the procedure is $T$,
a cutoff value $\Lb \in \mathbb{R}$ and an integer $n \in
\mathbb{N}$. The procedure returns the number of vertices in $T$
whose value is at most $\Lb$.

It achieves that by exploring the tree $T$ in a depth first search manner,
starting at the root and turning back as soon as a node $w\in T$ such that
$\text{val}(w) > \Lb$ is encountered. Moreover, if the number of nodes whose
value is at most $\Lb$ exceeds $n$ during the search, the algorithm stops and
returns $n+1$.

The algorithm output is the following integer.
\[\Call{DFS}{T,\Lb,n} := \min\big\{\big|T_\Lb \big|, n+1\big\}. \]
Observe that the \Call{DFS}{} procedure allows us to check whether a node $w
\in T$ is a good node, i.e. whether $\val(w) \leq \opt^T(n)$. Indeed, $w$ is
good if and only if $\Call{DFS}{T, \val(w), n} \leq n$.

This algorithm visits only nodes in $T_{\Lb}$ or its direct descendants
and its running time is $O(n)$. The space complexity is $O(1)$, since the only values needed to be stored in memory are $\mathcal{L}$, $\text{val}{(v)}$, where $v$ is the root of the tree $T$, and a counter for the number of good values found so far.

\begin{figure}
\center
\includegraphics[width = 0.9\textwidth]{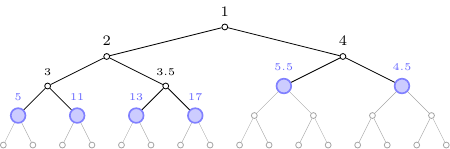}
\caption{An illustration of $R(T, \Lb_0)$ with $\Lb_0 = 4$. The number above each vertex is its value, the blue nodes are $R(T,\Lb_0)$, whereas the subtree above is $T_{\Lb_0}$.}
\label{fig_roots_pic_1}
\vspace{0.3cm}
\end{figure}

\paragraph{The procedure Roots} The procedure \Call{Roots}{} takes as input a
tree $T$ as well as an initial fixed lower bound $\Lb_0 \in
\mathbb{R}$ on the value of $\opt^T(n)$. We assume that the main algorithm has
already found all the nodes $w\in T$ satisfying $\text{val}(w) \leq
\Lb_0$. This means that the remaining values the main algorithm needs to find
in $T$ are all lying in the subtrees of the following nodes, that we call
the \emph{$\Lb_0$-roots of $T$}: \[R(T, \Lb_0) := \left\{r \in T
\setminus T_{\Lb_0} \; \big | \; r \text{ is a child of a node in }
T_{\Lb_0} \right\}\]
In other words, these are all the vertices in $T$ one level deeper in the tree than $T_{\Lb_0}$, see Figure \ref{fig_roots_pic_1} for an illustration.
In addition to that, the procedure takes two other parameters $\Lb, \Ub \in
\mathbb{R}$ as input, which correspond to (another) lower and upper bound on
the value of $\opt^T(n)$. These bounds $\Lb$ and $\Ub$ will be variables being
updated during the execution of the main algorithm, where $\Lb$ will be
increasing and $\Ub$ will be decreasing. More precisely, $\Lb$ will be the
largest value that the main algorithm has certified being at most $\opt^T(n)$,
whereas $\Ub$ will be the smallest value that the algorithm has certified being
at least that. A key observation is that these lower and upper bounds can allow
us to remove certain roots in $R(T,\Lb_0)$ from consideration, in the
sense that all the good values in that root's subtree will be certified to have
already been found. The only roots that the main algorithm needs to consider,
when $\Lb$ and $\Ub$ are given, are thus the following.

\begin{figure}
\center
\includegraphics[width = 0.9\textwidth]{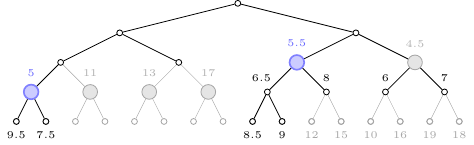}
\caption{An illustration of the \protect\Call{Roots}{} procedure with $\Lb_0 = 4, \Lb = 7 \text{ and } \Ub = 10$. Only two active roots remain, and are both colored in blue. The other roots are considered killed since all the good values have been found in their subtrees.}
\vspace{0.3cm}
\end{figure}

\begin{align}
\label{eq:roots_def}
\Call{Roots}{T,\Lb_0,\Lb,\Ub} := \left\{r \in R(T,\Lb_0) \mid \exists w \in T^{(r)} \text{ with val}(w) \in (\Lb,\Ub) \right\}
\end{align}

This subroutine can be implemented as follows. Run a depth first search starting at the root of $T$. Once a node $r \in T$ with $\text{val}(r) > \Lb_0$ is encountered, the subroutine marks that vertex $r$ as belonging to $R(T,\Lb_0)$. The depth first search continues deeper in the tree until finding a node $w \in T^{(r)}$ with $\text{val}(w) > \Lb$. At this point, if $\text{val}(w) < \Ub$, then the search directly returns to $r$ without exploring any additional nodes in $T^{(r)}$ and adds $r$ to the output. If however $\text{val}(w) \geq \Ub$, then the search continues exploring $T^{(r)}_{\Lb}$ (and its direct descendants) by trying to find a node $w$ with  $\text{val}(w) \in (\Lb,\Ub)$. In case the algorithm explores all of $T^{(r)}_{\Lb}$ with its direct descendants, and it turns out that no such node exists (i.e. every direct descendant $w$ of $T^{(r)}_{\Lb}$ satisfies $\text{val}(w) \geq \Ub$), then $r$ is not added to the output.

This procedure takes time $O(|T_\Lb|)$, i.e. proportional to the number of nodes in $T$ with value at most $\Lb$.
Since the procedure is called only on values $\Lb$ which are known to be good,
the time is bounded by $O(|T_\Lb|) = O(n)$.

In the main algorithm, we will only need this procedure in order to select a root from $\Call{Roots}{T,\Lb_0,\Lb,\Ub}$ uniformly at random, without having to store the whole list in memory. This can then be achieved in $O(1)$ space, since one then only needs to store $\text{val}(v), \Lb_0, \Lb$ and $\Ub$ in memory, where $v$ is the root of the tree $T$.

\paragraph{The procedure GoodValues} The procedure \Call{GoodValues}{} takes as input a tree $T$, a subtree $T^{(r)}$ for a node $r \in T$, a value $\Lb' \in \mathbb{R}_{\geq 0}$ and an integer $n \in \mathbb{N}$. The procedure then analyzes the set 
\[S:= \left\{\text{val}(w) \; \big| \; w \in T^{(r)}, \text{val}(w) \leq \Lb'\right\}\]
and outputs both the largest good value and the smallest bad value in that set, that we respectively call $\Lb$ and $\Ub$. If no bad values exist in $S$, the algorithm sets $\Ub = \infty$. Notice that this output determines, for each value in $S$, whether it is good or not. Indeed, any $\mathcal{V} \in S$ is good if and only if $\mathcal{V} \leq \Lb$, and is bad if and only if $\mathcal{V} \geq \Ub$.

The implementation is as follows. Start by initializing the variables $\Lb = -\infty$ and $\Ub = \Lb'$. These variables correspond to lower and upper bounds on $\opt^T(n)$. Loop through the values in 
\[S':= \left\{\text{val}(w) \mid w \in T^{(r)}, \;  \Lb < \text{val}(w) < \Ub\right\}\]
using a depth first search starting at $r$ and sample one value $\mathcal{V}$ uniformly randomly from that set. Check whether $\mathcal{V}$ is a good value by calling $\Call{DFS}{T, \mathcal{V}, n}$. If it is good, update $\Lb = \mathcal{V}$. If it is bad, update $\Ub = \mathcal{V}$. Continue this procedure until $S'$ is empty, i.e. $|S'| = 0$. If, at the end of the procedure, $\Lb = \Lb' = \Ub$, then set $\Ub = \infty$. The output is thus:
\[\Call{GoodValues}{T, T^{(r)}, \Lb',  n} := \{\Lb, \Ub\}\]
where 
\begin{align*}
\Lb &:= \max\left\{\mathcal{V} \in S \mid \mathcal{V} \leq \opt^T(n)\right\}, \\
\Ub &:= \min\left\{\mathcal{V} \in S \mid \mathcal{V} > \opt^T(n)\right\}.
 \end{align*}
 Sampling a value from $S'$ takes $O(|S|)$ time. Checking whether a sampled valued is good takes $O(n)$ time. In expectation, the number of updates before the set $S'$ is empty is $O(\log(|S|))$, leading to an expected total running time of $O((|S| + n)\log(|S|))$. As we will later see in the proof of \Cref{lemma_extend_proc}, we will only end up making calls $\Call{GoodValues}{T, T^{(r)},\Lb', n}$ with parameters $T^{(r)}$ and $\Lb'$ satisfying $\Call{DFS}{T^{(r)},\Lb'}=O(n)$. Since $|S|=\Call{DFS}{T^{(r)},\Lb'}$, this leads to an expected running time of $O(n\log(n))$.

The procedure can be implemented in $O(1)$ space, since the only values needed to be kept in memory are val($v$) (where $v$ is the root of the tree $T$), val($r$), $\Lb$, $\Ub$ and $\Lb'$, as well as the fact that every call to DFS also requires $O(1)$ space.

\subsection{The main algorithm}
We now present our main algorithm. This algorithm is named \Call{Select}{} and outputs the $n^{\text{th}}$ smallest value in the tree $T$. A procedure used in \Call{Select}{} is the \Call{Extend}{} algorithm, described below, which assumes that at least $n/2$ good nodes have already been found in the tree, and also outputs the $n^{\text{th}}$ smallest one.

\begin{algorithm}
	\caption{The \protect\Call{Select}{} procedure}
	\label{alg_select_procedure}
	\begin{algorithmic}[1]
		\State $\mathbf{Input: }$ $n \in \mathbb{N}$
		\State $\mathbf{Output: }$  $\opt(n)$, the $n^{\text{th}}$ smallest value in the heap $T$.
\Procedure{Select}{$n$}
		\State $k \gets 1$
		\State $\Lb \gets \text{val}(v)$ \Comment{$v$ is the root of the tree $T$}
\While{$k < n$}
		\If{$k < n/2$}
		\State $k' \gets 2k$
		\Else
		\State $k' \gets n$
		\EndIf
                \State $\Lb \gets \Call{Extend}{T, k', k, \Lb}$
		\State $k \gets k'$
		\EndWhile
\State \Return $\Lb$
		\EndProcedure
	\end{algorithmic}

\end{algorithm}

\begin{algorithm}
        \caption{The \protect\Call{Extend}{} procedure}
	\label{alg_extend_procedure}
	\begin{algorithmic}[1]
		\State \textbf{Input:} $T$: tree which is to be explored.
		\State \hspace{1.18cm}  $n \in \mathbb{N}$: total number of good values to be found, satisfying $n \geq 2$.
		\State \hspace{1.18cm}  $k \in \mathbb{N}$: number of good values already found, satisfying $k \geq n/2$.
		\State \hspace{1.18cm}  $\Lb_0 \in \mathbb{R}$: value satisfying $\Call{DFS}{T, \Lb_0, n} = k$.

		\State \textbf{Output: }   the $n^{\text{th}}$ smallest value in $T$. \vspace{0.2cm}

		\Procedure{Extend}{$T$, $n$, $k$, $\Lb_0$}
		\State $\Lb \gets \Lb_0$ \label{line:after_k_change}
		\State $\Ub \gets \infty$
		\While{$k<n$} \label{line:outer-while}
		\State $r\gets $ random element from \Call{Roots}{$T$, $\Lb_0$, $\Lb$, $\Ub$} \label{line:random_root}
		\vspace{0.2cm}
		\State $\Lb'\gets \max(\Lb,\text{val}(r))$
		\State $k'\gets$ \Call{DFS}{$T$, $\Lb'$, $n$} \Comment{count the number of values $\leq \Lb'$ in $T$} \label{line:first_dfs}
		\State $c\gets $ \Call{DFS}{$T^{(r)}$, $\Lb'$, $n$} \Comment{counting the number of values $\leq \Lb'$ in $T^{(r)}$} \label{line:second_dfs}
		\State $c'\gets \min(n-k'+c, 2c)$ \Comment{increase the number of values to be found in $T^{(r)}$} \label{line:first_c_prime}
		\While{$k' < n$} \Comment{loop until it is certified that $\opt^{T}(n) \leq \Lb'$}
		\State $\Lb'\gets$ \Call{Extend}{$T^{(r)}$, $c'$, $c$, $\Lb'$}  \label{line:recurse}
		\State $k'\gets$ \Call{DFS}{$T$, $\Lb'$, $n$} \label{line:third_dfs}
		\State $c\gets  c'$
		\State $c'\gets \min(n-k'+c, 2c)$ \label{line:end_inner_loop}
		\EndWhile
		\State $\tilde{\Lb}, \tilde{\Ub} \gets \Call{GoodValues}{T, T^{(r)}, \Lb', n}$ \Comment{find the good values in $T^{(r)}$} \label{line:goodvalues}
		\State $\Lb \gets \max(\Lb, \tilde{\Lb})$
		\State $\Ub \gets \min(\Ub, \tilde{\Ub})$
                \State $k\gets$ \Call{DFS}{$T$, $\Lb$, $n$} \Comment{compute the number of good values found in $T$\hspace{-1cm}} \label{line:fourth_dfs}
		\EndWhile
		\State \Return $\Lb$
		\EndProcedure
	\end{algorithmic}
\end{algorithm}

Let us describe a few invariants from the \Call{Extend}{} procedure.
\begin{itemize}
\item $\Lb$ and $\Ub$ are respectively lower and upper bounds on $\opt^{T}(n)$ during the whole execution of the procedure. More precisely, $\Lb \leq \opt^{T}(n)$ and $\Ub > \opt^{T}(n)$ at any point, and hence $\Lb$ is good and $\Ub$ is bad. The integer $k$ counts the number of values $\leq \Lb$ in the full tree $T$.
\item No root can be randomly selected twice. This is ruled out by the updated values of $\Lb$ and $\Ub$, and the proof can be found in \Cref{thm:correctness}.
\item After an iteration of the inner while loop, $\Lb'$ is set to the $c^{\text{th}}$ smallest value in $T^{(r)}$. The variable $c'$ then corresponds to the next value we would like to find in $T^{(r)}$ if we were to continue the search. Note that $c' \leq 2c$, enforcing that the recursive call to $\Call{Extend}{}$ satisfies its precondition, and that $c' \leq n - (k'-c)$ implies that $(k'-c) + c' \leq n$,
    which implies that the recursive subcall will not spend time searching for a value that is known in advance to be bad.
\item From the definition of $k'$ and $c$ one can see that $k'\geq k+ c$. Combined with the previous invariant, we see that $c'\leq n-k$.
\item $k'$ always counts the number of values $\leq \Lb'$ in the full tree $T$. It is important to observe that this is a global parameter, and does not only count values below the current root. Moreover, $k' \geq n$ implies that we can stop searching below the current root, since it is guaranteed that all good values in $T^{(r)}$ have been found, i.e., $\Lb'$ is larger than all the good values in $T^{(r)}$.
\end{itemize}
\FloatBarrier

\subsection{Proof of correctness}
\begin{theorem}\label{thm:correctness}
At the end of the execution of Algorithm \ref{alg_select_procedure}, $\Lb$ is set to the $n^{\text{th}}$ smallest value in $T$.
Moreover, the algorithm is guaranteed to terminate.
\end{theorem}

\begin{proof}
We show $\Lb = \opt^{T}(n)$ holds at the end of Algorithm \ref{alg_extend_procedure}, i.e. the \Call{Extend}{} procedure. Correctness of Algorithm \ref{alg_select_procedure}, i.e. the \Call{Select}{} procedure, then clearly follows from that. First, notice that $\Lb$ is always set to the first output of the procedure \Call{GoodValues}{}, which is always the value of a good node in $T$, implying
\[\Lb \leq \opt^{T}(n)\]
at any point during the execution of the algorithm.
Since the outer while loop ends when at least $n$ good nodes in $T$ have value at most $\Lb$, we get
\[\Lb \geq \opt^{T}(n),\]
which implies that when the algorithm terminates it does so with the correct value.

It remains to prove that the algorithm terminates.
We observe that every recursive call $\Lb' \gets \Call{Extend}{T^{(r)},c',c,\Lb'}$ strictly
increases the value of $\Lb'$, ensuring that at least one extra value in $T$ is under
the increased value. This implies that $k'$ strictly increases every iteration of the
inner while loop, thus ensuring that this loop terminates.

To see that the outer loop terminates, observe that after each iteration the set \Call{Roots}{$T, \Lb_0, \Lb, \Ub$} shrinks by at least one element. As soon as this set is empty, there will be no more roots with unexplored good values in their subtrees, so $k=n$ and the algorithm terminates.

\end{proof}

\subsection{Running time analysis}
In order to prove a $O(n \log(n^3))$ running time bound for the $\Call{Select}{n}$ procedure, we will show that the running time of the $\Call{Extend}{}$ procedure with parameters $n$ and $k$ is $O((n-k)\log(n)^3) + O(n\log(n)^2)$.

The main challenge in analyzing the running time of $\Call{Extend}{}$ is in dealing with the cost of the recursive subcalls on line~\ref{line:recurse}.
For this we rely on an important idea, formalized in Lemma \ref{lem:exp_sum}, stating that if the parent call with parameters $n$ and $k$ makes $z \in \mathbb{N}$ recursive calls with parameters $(n_1,k_1),\ldots, (n_z,k_z)$, then $\sum_{i=1}^z (n_i-k_i)\leq n-k$ in expectation over the random choices of the algorithm.

A second insight is that the outermost while loop on line~\ref{line:outer-while} is executed at most $O(\log(n))$ times in expectation, which is shown in Lemma \ref{lem:exp_iterations}. The first lemma allows to show that the running time of the $\Call{Extend}{}$ procedure on the recursive part is $O((n-k)\log(n)^3)$, through an induction proof. The second lemma helps to show that the running time of the $\Call{Extend}{}$ procedure on the non-recursive part is $O(n \log(n)^2)$. The running time analysis of $\Call{Extend}{}$ is formally done in Lemma \ref{lemma_extend_proc}. Finally, the running time of $O(n \log(n^3))$ for the $\Call{Select}{n}$ procedure then follows in Theorem \ref{thm:select}.

Let us now prove these claims. We first show that the expectation of $\sum_{i=1}^z(n_i-k_i)$ is bounded by $n-k$.

\begin{lemma}
		\label{lem:exp_sum}
		Let $z$ be the number of recursive calls with $k\geq 1$ that are done in the main loop of \Call{Extend}{$T$, $n$, $k$, $\Lb_0$}. For every $i\in [z]$, let $n_i$ and $k_i$ be the values that are given as second and third parameters to the $i$th such subcall. It holds that:
		\begin{align*}
			\Exp\left[\sum_{i=1}^z (n_i-k_i)\right]\leq n-k.
		\end{align*}
\end{lemma}
\begin{proof}
For simplicity of notation, let us denote the set of roots at the beginning of the execution of the algorithm by
$\mathcal{R} := \Call{Roots}{T, \Lb_0, \Lb, \Ub}$,
where $\Lb = \Lb_0$ and $\Ub = \infty$ at initialization. An important observation is that, once a root $r \in \mathcal{R}$ is randomly selected on line~\ref{line:random_root}, all the recursive calls under it (i.e. with its subtree $T^{(r)}$ as first parameter) on line~\ref{line:recurse} are consecutive. The last such recursive call ensures that all the good values in $T^{(r)}$ are found and sets $\Lb$ and $\Ub$ to respectively be the largest good value and smallest bad value in $T^{(r)}$. From then on, this root leaves the updated set $\Call{Roots}{T, \Lb_0, \Lb, \Ub}$  by \eqref{eq:roots_def} and will thus never be again considered in the random choice on line~\ref{line:random_root}. For every $r \in \mathcal{R}$, let us define the set:
\begin{align*}
C(r) = \Big\{i \in [z] \text{ s.t. the }i\text{th recursive call is under root $r$} \Big\}
\end{align*}
and let us denote by $S_r \in \mathbb{N}$ the total number of good values in its subtree $T^{(r)}$. Our goal is to show that:
\begin{align}
\label{eq_goal}
\Exp\Big[\sum_{i \in C(r)} (n_i-k_i)\Big]\leq S_r \qquad \forall r \in \mathcal{R}.
\end{align}
Clearly, this would imply the lemma, since the total number of good values to be found is $\sum_{r \in \mathcal{R}}S_r = n-k$. For convenience, we define this number to be $p: = n-k$. We now order the good values to be found and denote them as follows: $\mathcal{V}_1 < \mathcal{V}_2 < \dots < \mathcal{V}_p$. Each value $\mathcal{V}_k$ is to be found in the subtree of a certain root that we denote by $r(\mathcal{V}_k) \in \mathcal{R}$.

We first show that the claim \eqref{eq_goal} holds for any root $r \in \mathcal{R}$ such that $r \neq r(\mathcal{V}_p)$. Let us thus fix such a root $r \neq r(\mathcal{V}_p)$. The key observation is that, since the random choice on line~\ref{line:random_root} is uniform, and since $r(\mathcal{V}_p)$ will always be among the active roots, the subtree of the root $r(\mathcal{V}_p)$ will be explored before the subtree of root $r$ with probability at least a half. In that case, no recursive calls will be made under root $r$. This holds since the updated values $\Lb$ and $\Ub$ after the iteration of $r(\mathcal{V}_p)$ ensure that $r$ leaves $\Call{Roots}{T, \Lb_0, \Lb, \Ub}$ by $\eqref{eq:roots_def}$ and is thus not considered in the random choice in later iterations. If the root $r$ is however considered before $r(\mathcal{V}_p)$, which happens with probability at most a half, then $\sum_{i \in C(r)} (n_i-k_i) \leq 2 S_r$, since the sum is telescoping and the parameters $k_i$ and $n_i$ at most double at each step on line~\ref{line:end_inner_loop} until all good values in $T^{(r)}$ are found. Hence, we get that
\begin{align}
\Exp\Big[\sum_{i \in C(r)} (n_i-k_i)\Big] \leq \frac{1}{2} \; 0 + \frac{1}{2} \; 2 S_r \leq S_r.
\end{align}

It remains to show that claim \eqref{eq_goal} holds for the root $r(\mathcal{V}_p)$ under which the largest good value lies. In that case, let us denote by $\mathcal{V}_j$ the largest good value lying in a subtree of a different root $r(\mathcal{V}_j) \neq r(\mathcal{V}_p)$. We also denote by $\{r(\mathcal{V}_j) \prec r(\mathcal{V}_p)\}$ the probabilistic event that $r(\mathcal{V}_j)$ is considered before $r(\mathcal{V}_p)$ in the random choices of the algorithm. By our choice of $\mathcal{V}_j$ and $\mathcal{V}_p$, this event happens with probability exactly a half. Moreover, if this event happens, all the good values outside of $T^{(r(\mathcal{V}_p))}$ will have been found after exploring $T^{(r(\mathcal{V}_j))}$. This means that, when the algorithm considers $r(\mathcal{V}_p)$, it knows that there remain at most $p-j$ values to be found. That is, we will have $C(r(\mathcal{V}_p))=\{t,\ldots, z\}$ for some $t$, such that $k_t\geq S_{r(\mathcal{V}_p)}-(p-j)$ and $n_z \leq S_{r(\mathcal{V}_p)}$, leading to
\begin{align}
\label{eq_last_root_condition}
\Exp\Big[\sum_{i \in C(r(\mathcal{V}_p))} (n_i-k_i) \; \big| \;  r(\mathcal{V}_j) \prec r(\mathcal{V}_p) \Big] \leq S_{r(\mathcal{V}_p)} - \left(S_{r(\mathcal{V}_p)}-(p-j)\right) = p - j,
\end{align}
where we have again used the fact that the sum is telescoping.

We now consider the event $\{r(\mathcal{V}_p) \prec r(\mathcal{V}_j)\}$ and distinguish two cases. Suppose that the penultimate call $i \in C(r(\mathcal{V}_p))$ finds a good value which is bigger than $\mathcal{V}_{j}$. By a similar argument as above, the algorithm does not double in the last step, but truncates due to line~\ref{line:end_inner_loop}, meaning that $\sum_{i \in C(r(\mathcal{V}_p))} (n_i-k_i) = S_{r(\mathcal{V}_p)}$ holds in this case. Combining this with \eqref{eq_last_root_condition} and using the fact that the last $p-j$ values are under root $r(\mathcal{V}_p)$, we get
\[\Exp\Big[\sum_{i \in C(r(\mathcal{V}_p))} (n_i-k_i)\Big] \leq \frac{1}{2}(p-j) + \frac{1}{2}S_{r(\mathcal{V}_p)} \leq S_{r(\mathcal{V}_p)}. \]
Suppose now that the penultimate call $i \in C(r(\mathcal{V}_p))$ finds a good value which is smaller than $\mathcal{V}_{j}$. This means that the number of good values found in $T^{(r(\mathcal{V}_p))}$ is at most $S_{r(\mathcal{V}_p)} - (p-j)$ at that point. The last call $i \in C(r(\mathcal{V}_p))$ then doubles the parameters, meaning that $\sum_{i \in C(r(\mathcal{V}_p))} (n_i-k_i) \leq 2 \: (S_{r(\mathcal{V}_p)} - (p-j) )$ holds, due to the fact that the sum is telescoping. Combining this with \eqref{eq_last_root_condition} leads to
\[\Exp\Big[\sum_{i \in C(r(\mathcal{V}_p))} (n_i-k_i)\Big] \leq \frac{1}{2}(p-j) + S_{r(\mathcal{V}_p)} - (p-j) \leq S_{r(\mathcal{V}_p)}.\]
\end{proof}

We now bound the expected number of iterations of the outermost while-loop.
\begin{lemma}
		\label{lem:exp_iterations}
	The expected number of times that the outermost while-loop (at line \ref{line:outer-while}) is executed by the procedure \Call{Extend}{} is at most $O(\log (n))$.
\end{lemma}

\begin{proof}
Let $r_1,\ldots, r_m$ denote the roots returned by $\Call{Roots}{T, \Lb_0, \Lb_0, \infty}$. For $j\in [m]$, let $\ell_j$ and $u_j$ respectively denote the largest good value and the smallest non-good value under root $r_j$. Let $A_\ell(\Lb):= \{r_j: \ell_j > \Lb \}$ and $A_u(\Ub):=\{r_j: u_j < \Ub \}$. Observe that $\Call{Roots}{T, \Lb_0, \Lb, \Ub}=A_\ell(\Lb)\cup A_u(\Ub)$ for any $\Lb \leq \Ub$.

Let $\Lb_i$ and $\Ub_i$ denote the values of $\Lb$ and $\Ub$ at the start of the $i$th iteration. After an iteration $i$ in which root $r_j$ was selected, the algorithm updates $\Lb$ and $\Ub$ such that $\Lb_{i+1}=\max(\Lb, \ell_j)$ and $\Ub_{i+1}=\min(\Ub, u_j)$. Observe that $\Lb_i$ is nondecreasing and that $\Ub_i$ is nonincreasing.

We will now show that if a root from $A_\ell(\Lb_i)$ is selected in iteration $i$, then the expected size of $A_\ell(\Lb_{i+1})$ is at most half that of $A_\ell(\Lb_i)$. This will imply that in expectation only $\log(n)$ iterations are needed to make $|A_\ell(\Lb)|=1$.

Let $\mathcal{F}_i$ be the filtration containing all information up until iteration $i$.
Let $X_i$ be a random variable denoting the value of $|A_\ell(\Lb_i)|$. Let $(s_k)_{k \geq 1}$ be the subsequence consisting of iteration indices $i$ in which a root from $A_\ell(\Lb_i)$ is selected. Because roots are selected uniformly at random, we have $\Exp[X_{s_{k+1}} \mid \mathcal{F}_{s_{k}}]\leq \frac12 X_{s_k}$.

Let $Y_i=\max(\log(X_i), 0)$. Note that when $Y_{s_{k}}\geq 1$, we have $\Exp[Y_{s_{k+1}} \mid \mathcal{F}_{s_{k}}]=\Exp[\log(X_{s_{k+1}}) \mid \mathcal{F}_{s_{k}}]\leq \log(\Exp[X_{s_{k+1}}\mid \mathcal{F}_{s_{k}}])\leq Y_{s_{k}} -1$. Let $\tau$ be the smallest $k$ such that $Y_{s_{k}}=0$. Note that $\tau$ is the number of iterations $i$ in which a root from $A_\ell(\Lb_i)$ is selected, and hence $\tau \leq n$. The sequence $(Y_{s_k} + k)_{k=1,\ldots, \tau}$ is therefore a supermartingale and $\tau$ is a stopping time.  By the martingale stopping theorem \cite[Theorem 12.2]{mitzenmacher_probability_2005-2}, we have $\Exp[\tau]=\Exp[Y_{s_\tau} +\tau ]\leq \Exp[Y_{s_1}+1]=\log(m)+1$.

Now we have shown that in expectation at most $\log(m)+1$ iterations $i$ are needed in which roots from $A_\ell(\Lb_i)$ are considered.
The same argument can be repeated for  $A_u(\Ub)$.
Since in every iteration a root from $A_\ell(\Lb)$ or $A_u(\Ub)$ is selected, the expected total number of iterations is at most $2\log(m)+2$. This directly implies the lemma as $m\leq |T_\Lb| + 1 \leq n + 1$.
\end{proof}
We are now able to prove the running time bound for the \Call{Extend}{} procedure.
\begin{lemma}
\label{lemma_extend_proc}
Let $R(T, n,k)$ denote the running time of a call \Call{Extend}{$T$, $n$, $k$, $\Lb_0$}. Then
there exists $C > 0$ such that
\begin{align*}
\Exp[R(T, n,k)]\leq 5C(n-k)\log(n)^3+Cn\log(n)^2.
\end{align*}
\end{lemma}
\begin{proof}
We will prove this with induction on $r:=\lceil \log(n) \rceil$. For $r=1$, we have $n\leq 2$. In this case $R$ is constant, proving our induction base.

Now consider a call \Call{Extend}{$T$, $n$, $k$, $\Lb_0$} and assume the induction claim is true when $\lceil \log(n) \rceil \leq r-1$. Let $p$ be the number of iterations of the outer-most while-loop that are executed.

We will first consider the running time induced by the base call itself, excluding any recursive subcalls. Note that all of this running time is incurred by the calls to the procedures \Call{DFS}{}, \Call{Roots}{} and \Call{GoodValues}{}, plus the cost of moving to the corresponding node before each of these calls. In the base call, the procedure will only move between nodes that are among the ones with the $n$ smallest values, or the nodes directly below them. For this reason, we can upper bound the cost of each movement action by a running time of $O(n)$.
\begin{itemize}
\item On line \ref{line:first_dfs},  \ref{line:second_dfs}, \ref{line:fourth_dfs} each call \Call{DFS}{} incurs a running time of at most $O(n)$. Each of these lines will be executed $p$ times, incurring a total running time of $O(pn)$.
\item On line \ref{line:third_dfs} each call \Call{DFS}{$T$, $\Lb'$, $n$} incurs a running time of at most $O(n)$. The code will be executed $O(p\log(n))$ times since $c'$ doubles after every iteration of the inner loop and never grows larger than $n$, thus incurring a total running time of $O(pn\log(n))$.
\item The arguments $T^{(r)}$ and $\Lb'$ of the call to \Call{GoodValues}{} on line \ref{line:goodvalues} satisfy $\Call{DFS}{T^{(r)},\Lb'}=c\leq c'\leq n$. Hence, the running time of this procedure is $O(n\log (n))$ time. The line is executed at most $p$ times, so the total running time incurred is $O(pn\log(n))$.
\end{itemize}
Adding up all the running times listed before, we see that the total running time incurred by the non-recursive part is $O(pn\log(n))$.
By \Cref{lem:exp_iterations}, $\Exp[p]\leq \log(n)$.
Hence,
we can choose $C$ such that
the expected running time of the non-recursive part is bounded by
\[Cn\log(n)^2.\]

Now we move on to the recursive part of the algorithm. All calls to
\Call{Extend}{$T$, $n$, $k$, $\Lb_0$} with $k=0$ will have $n=1$, so each of
these calls takes only $O(1)$ time. Hence we can safely ignore these calls.

Let $z$ be the number of of recursive calls to \Call{Extend}{$T$, $n$, $k$,
$\Lb_0$} that are done from the base call with $k\geq 1$.  Let $T_i$, $k_i$,
$n_i$ for $i\in [z]$ be the arguments of these function calls. Note that
$n/2 \geq n - k \geq  n_i\geq 2$ for all $i$.  By the induction hypothesis, the expectation of the
recursive part of the running time is:
\begin{align*}
\Exp\left[\sum_{i=1}^zR(T_i,n_i,k_i)\right]&\leq \Exp\left[ \sum_{i=1}^z 5C (n_i-k_i)\log(n_i)^3+Cn_i\log(n_i)^2 \right]\\
&\leq 5C\log(n/2)^3\Exp\left[\sum_{i=1}^z (n_i-k_i)\right]+C \log(n/2)^2 \Exp\left[\sum_{i=1}^z n_i \right]\\
&\leq 5C(\log(n)-1)\log(n)^2\Exp\left[\sum_{i=1}^z (n_i-k_i)\right]+C \log(n)^2 \Exp\left[\sum_{i=1}^z n_i\right]\\
&\leq 5C(\log(n)-1)\log(n)^2(n-k)+5C \log(n)^2 (n-k)\\
&\leq 5C(n-k)\log(n)^3.
\end{align*}
Here we used \Cref{lem:exp_sum} as well as the fact that
$\sum_{i=1}^z n_i\leq 4(n-k)$.
To see why the latter inequality is true, consider an arbitrary root
$r$ that has $S_r$ values under it that are good (with respect to the base call). Now
$\sum_{i=1}^{z}\mathbf{1}_{\{T_i = T^{(r)}\}} n_i \leq \sum_{i=2}^{\lceil \log(S_r+1) \rceil}2^i\leq 2^{\lceil \log(S_r+1) \rceil +1}\leq 4S_r$.
In total there are $n-k$ good values under the roots, and hence
$\sum_{i=1}^{z}n_i \leq 4(n-k)$.

Adding the expected running time of the recursive and the non-recursive part, we see that \[\Exp[R(T,n,k)]\leq 5C(n-k)\log(n)^3+Cn\log(n)^2.\]
\end{proof}

This now implies the desired running time for the procedure $\Call{Select}{}$.

\begin{theorem}
\label{thm:select}
The procedure $\Call{Select}{n}$ runs in expected $O(n \log(n)^3)$ time.
\end{theorem}
\begin{proof}
The key idea is that $\Call{Select}{}$ calls $\Call{Extend}{T, k', k, \Lb}$ at most $\lceil{\log(n)}\rceil$ times with parameters $(k', k) = (2^i, 2^{i-1})$ for $i \in \{1,\dots, \lceil{\log(n)}\rceil\}$.
By Lemma \ref{lemma_extend_proc}, the running time of $\Call{Select}{}$ can thus be upper bounded by
\begin{align*}
\sum_{i=1}^{\lceil{\log(n)} \rceil}\Exp[R(T, 2^i,2^{i-1})] &\leq 5C \log(n)^3 \sum_{i=1}^{\lceil{\log(n)} \rceil} (2^i - 2^{i-1}) + \sum_{i=1}^{\lceil{\log(n)} \rceil} C n \log(n)^2 \\ &= O(n \log(n)^ 3).
\end{align*}
\end{proof}

\subsection{Space complexity analysis}
We prove in this section the space complexity of our main algorithm.

\begin{theorem}
The procedure $\Call{Select}{n}$ runs in $O(\log(n))$ space.
\end{theorem}
\begin{proof}
Observe that it is enough to prove that the statement holds for $\Call{Extend}{T, n, k, \Lb}$ with $k \geq n/2$, since the memory can be freed up (only keeping the returned value in memory) after every call to $\Call{Extend}{}$ in the $\Call{Select}{n}$ algorithm.

Moreover, observe that the subroutines $\Call{DFS}{}$, $\Call{Roots}{}$ and $\Call{GoodValues}{}$ all require $O(1)$ memory, as argued in their respective analyses.  Any call $\Call{Extend}{T, n, k, \Lb}$ only makes recursive calls $\Call{Extend}{T^{(r)}, \hat{n}, \hat{k}, \hat{\Lb}}$ with $1\leq \hat{n}\leq n-k\leq \frac12 n$. So the depth of the recursion is at most $\log(n)$, and the space complexity of the algorithm is $O(\log(n))$.
\end{proof}

\section{Lower bound}
\label{sec:lowerbound}
No lower bound is known for the running time of the selection problem on explorable heaps. However, we will show that any (randomized) algorithm with space complexity at most $s$, has a running time of at least $\Omega(n\log_{s}(n))$. Somewhat surprisingly, the tree that is used for the lower bound construction is very simple: a root with two trails of length $O(n)$ attached to it.

We will make use of a variant of the communication complexity model. In this model a totally ordered set $W$ is given, which is partitioned into $(S_A,S_B)$. There are two agents $A$ and $B$, that have access to the sets of values in $S_A$ and $S_B$ respectively. We have $|S_A|=n+1$ and $|S_B|=n$. Assume that all values $S_A$ and $S_B$ are different. Now consider the problem where player $A$ wants to compute the median, that is the $(n+1)$th smallest value of $W$.

Because the players only have access to their own values, they need to communicate. For this purpose they use a protocol, that can consist of multiple rounds. In every odd round, player $A$ can do computations and send units of information to player $B$. In every even round, player $B$ does computations and sends information to player $A$. We assume that sending one value from $S_A$ or $S_B$ takes up one \emph{unit of information}. Furthermore, we assume that, except for comparisons, no operations can be performed on the values. For example, the algorithm cannot do addition or multiplication on the values.

We will now reduce the median computation problem to the explorable heap selection problem.

\begin{lemma}
	If there is a randomized algorithm that solves $\opt(3n)$ in $f(n)n$ time and $g$ space, then there is a randomized protocol for median computation that uses $f(n)/2$ rounds in each of which at most $g$ units of information are sent.
        \label{lem:reduce_communication}
\end{lemma}
\begin{proof}
	Consider arbitrary sets $S_A$ and $S_B$ with $|S_A|=n+1$ and $|S_B|=n$ and $S_A\cap S_B=\emptyset$. Introduce a new element $O$, such that $O< x$ for all $x\in S_A\cup S_B$. Let $M_A$ and $M_B$ be two sets with $|M_A|=|M_B|+1=n$ and $O<y<x$ for all $y\in M_A\cup M_B$ and $x\in S_A\cup S_B$.

Let us write $S_{A}=\{a_1,\ldots, a_{n+1}\}$. Now consider a subtree for which the root node has value $a_1$. For every $i \in \{1, \dots, n\}$, let every node that has value $a_i$ have a child with value $a_{i+1}$ and another child with some value that is larger than any value in $S_A\cup S_B\cup M_A\cup M_B$. We will call this a \emph{trail} of $S_A$.
	\begin{figure}[h!]
		\centering
		\begin{tikzpicture}
			\node {O}
			child {node {$M_A$} child {node {$S_A$}} child {edge from parent[draw=none]}}
			child {node {$M_B$}  child {edge from parent[draw=none]} child {node {$S_B$}}};
		\end{tikzpicture}
	\end{figure}

	Now we will construct a labeled tree in the following way: create a tree with a root node of value $O$. Attach a trail of $M_A$ as the left child of this root and a trail of $M_B$ as the right child.
	Attach a trail of $S_A$ as a child of the largest node in $M_A$ and do the same for a trail of $M_B$ under the largest node of $S_B$. The resulting tree will now look as shown in the above picture.

	Observe that the $3n$th smallest value in this tree is the median of $S_A\cup S_B$. Now we can view the selection algorithm as an algorithm for median computation if we consider moving between $S_A$ and $S_B$ in the tree as sending the $g$ units of information that are in memory to the other player. Because moving between the two sets takes at least $2n$ steps, the number of rounds of rounds in the corresponding communication protocol is at most $\frac{f(n)n}{2n}=f(n)/2$, proving the statement.
\end{proof}

We now move on to proving a lower bound for the median computation problem. The following lemma will play a key part in the proof.
\begin{lemma}
	Let $S\subseteq [n]$ be a randomly distributed subset of $[n]$ with size $|S|\leq k\leq n$. Then for $\ell\leq \frac{n}{8k}$ there exists a length-$\ell$ interval $\subseteq [n]$ (i.e.\ $I=\{i,i+1,\ldots, i+\ell-1\}$) such that:
	$\Pr[S\cap I\neq \emptyset] \leq \frac14$.
	\label{lem:interval}
\end{lemma}
\begin{proof}
	Let $\mathcal{I}_\ell$ be the set of length-$\ell$ intervals in $[n]$. We have $|\mathcal{I}_\ell|=n-\ell+1$. Observe that any value in $[n]$ is contained in at most $\ell$ elements of $\mathcal{I}_\ell$. Hence, for any set $S$ of size at most $k$, there are at most $k\cdot \ell$ elements of $\mathcal{I}_\ell$ that contain any of the elements of $S$. That is: $|\{I\in \mathcal{I}_\ell: I\cap S\neq \emptyset \}|\leq k\cdot \ell$. This implies that for a randomly distributed set $S\subseteq [n]$ we also have:
	\begin{align*}
		\sum_{I\in\mathcal{I}_\ell} \Pr_S[I\cap S\neq \emptyset]&=\sum_{I\in \mathcal{I}_\ell} \Exp_S[\mathbf{1}_{I\cap S \neq \emptyset}]=\Exp_S\left[\sum_{I\in \mathcal{I}_\ell} \mathbf{1}_{I\cap S \neq \emptyset}\right]\\&=\Exp_S[|\{I\in \mathcal{I}_\ell: I\cap S\neq \emptyset \}|]\leq k\cdot \ell.
	\end{align*}
	So there must be an $I\in \mathcal{I}_\ell$ with:
	\begin{align*}
		\Pr_S[I\cap S\neq \emptyset]\leq \frac{k\cdot \ell}{|\mathcal{I}_\ell|} =  \frac{k\cdot \ell}{n-\ell+1}\leq \frac{k\cdot \frac{n}{8k}}{\frac12 n}=\frac14.
	\end{align*}
\end{proof}

\begin{theorem}
\label{lower_bound_median}
Any randomized protocol for median computation that sends at most $g$ units of info per round, takes at least $\Omega(\log_{g+1}(n))$ rounds in expectation.
\end{theorem}

\begin{proof}
	We will instead prove the following result for a symmetric version of median computation, because this makes the proof a bit easier. In this setting, we have  $|S_A|=|S_B|=n$ and the objective is to find both the $n$th and the $(n+1)$th smallest element of $S_A\cup S_B$. We will call the set consisting of these two values the \emph{$2$-median} of $S_A\cup S_B$ and we will denote it by $\tmedian(S_A\cup S_B)$. Because this problem can be trivially solved by appending two rounds to any median-computation algorithm, proving a lower bound for this case is sufficient.
	
	Let $g'= g+1$. We can assume that $g\geq 1$, and hence $g'\geq 2$.
	We will prove with induction on $n$ that the expected number of rounds to compute the median is at least $\frac{1}{10} \log_{g'}(n)-9$. For $n< 2^8(g')^2$, this is trivial. Now let $n\geq 2^8(g')^2$. Assume that the claim is true for values strictly smaller than $n$. We will now prove the claim for $n$.

	Consider an arbitrary randomized algorithm. Let $V_i \subseteq [n]$ be the set of indices of the values that are emitted during round $i$ by one of the two players. Observe that the distribution of the set $V_1$ does not depend on the input, because player $A$ only has access to his own set of $n$ values that he can compare to each other. Order the values in $S_A$ by increasing order of their values $x_1,\ldots, x_{n}$. Order the values of $S_B$ in decreasing order as $y_1,\ldots, y_{n}$. We now describe below how the relative ordering of the $x_i$'s with respect to the $y_i$'s is decided adversarially.

	Let $\ell=\lfloor{\frac{n}{8g}}\rfloor$. From \Cref{lem:interval} it follows that there exists an interval $I=\{a,\ldots, a+\ell-1\}\subseteq [n]$ such that $\Pr[V_1\cap I\neq \emptyset]\leq \frac14$.
	Now let $L=\{1,\ldots, a-1\}$ and $U=\{a+\ell,\ldots, n\}$. Observe that $\{L,I,U\}$ forms a partition of $[n]$.
	We now order the values in the sets such that we have  $y_{u}<x_l<y_i<x_u<y_{l}$ for all $l\in L, u\in U, i\in I$. Note that this fixes the ordinal index of any element in $S_A\cup S_B$, except for the elements $x_i$ and $y_i$ for $i\in I$.

	Condition on the event that $I\cap V_1=\emptyset$. Observe that in this case, the results of all comparisons that player 2 can do in the second round have been fixed. Hence, $V_2$ will be a random subset of $[n]$, whose distribution will not depend on the order of the values $x_a,\ldots, x_{a+\ell-1}$ with respect to $y_1,\ldots, y_{n}$.

	We now do a similar argument for the second player. Let $\ell'=\lfloor \frac{\ell}{8g} \rfloor$. From \Cref{lem:interval}, there exists an interval $I'=\{a',\ldots, a'+\ell'-1 \}\subseteq I$ such that
	$\Pr[I'\cap V_2\neq \emptyset \mid I\cap V_1=\emptyset] \leq \frac14$. Define $L'=\{a,\ldots, a'-1\}$ and $U'=\{a'+\ell',\ldots, a+\ell-1\}$. Observe that $\{L',I',U'\}$ forms a partition of $I$.
	We now order the values in the sets such that we have  $y_{u}<x_l<y_i<x_u<y_{l}$ for all $l\in L', u\in U', i\in I'$. Note that we have now fixed the ordinal index of any element in $S_A\cup S_B$, except for the elements $x_i$ and $y_i$ for $i\in I'$.

	Because $I'\subseteq I$, we have \[\Pr[I'\cap (V_1 \cup V_2 )\neq \emptyset] \leq \Pr[I\cap V_1\neq \emptyset]+\Pr[I'\cap V_2\neq \emptyset \mid S\cap V_1=\emptyset]\leq \frac14+\frac14=\frac12.\]

	Now, let $R$ be the number of rounds that the algorithm takes and define $S_A'=\{x_i: i\in I'\}$ and $S_B'=\{y_{i}: i\in I'\}$.
	Observe that $\tmedian(S_A\cup S_B)=\tmedian(S_A'\cup S_B')$. So the algorithm can now be seen as an algorithm to compute the $2$-median of $S_A'\cup S_B'$.
	Let $R'$ be the number of rounds in which elements from the set $S_A'\cup S_B'$ are transmitted. With probability $\phi:=Pr[I'\cap (V_1 \cup V_2 ) =\emptyset ]\geq \frac12$, no information about $S_A'$ and $S_B'$ is transmitted in the first two rounds, meaning that \[\Exp[R']\leq\phi\Exp[R-2]+(1-\phi)\Exp[R]= \Exp[R]-2\phi\leq \Exp[R]-1.\]
	
	Moreover, by our induction hypothesis it follows that $R'$ satisfies:
	\begin{align*}
		\Exp[R']&\geq \frac{1}{10} \log_{g'}(|S_B'|)-9=
		\frac1{10}\log_{g'}(\ell')-9\geq
		\frac1{10}\log_{g'}\left(\frac{n}{(8g)^2}-2\right)-9
		\\&\geq
		\frac{1}{10}(\log_{g'}(n)-2\log_{g'}(8g') -2)-9 \geq \frac{1}{10}\log_{g'}(n)-10.\end{align*}
		The second inequality follows from the definition of $\ell'$. The third inequality follows from the fact that $\log_{g'}(x-2)\geq \log_{g'}(x)-2$ for $x\geq 3$. The last inequality follows from $g' \geq 2$.
		Consequently, we get that $\Exp[R]\geq \Exp[R']+1\geq \frac{1}{10}\log_{g'}(n)-9$.
\end{proof}

Combining \Cref{lower_bound_median} and \Cref{lem:reduce_communication} now implies the following.
\newline

\begin{theorem}
The time complexity of any randomized algorithm for $\select(n)$ with at most $g$ units of storage is $\Omega(\log_{g+1}(n)n)$.
\end{theorem}
\bibliographystyle{alpha}
\bibliography{ref}

\begin{thebibliography}{DFKP04}

\bibitem[Ach09]{achterberg2009}
Tobias Achterberg.
\newblock {\em Constraint Integer Programming}.
\newblock PhD thesis, TU Berlin, 2009.

\bibitem[AG03]{theoryofsearch}
Steve Alpern and Shmuel Gal.
\newblock {\em The Theory of Search Games and Rendezvous}.
\newblock Number~55 in International Series in Operations Research \&
  Management Science. Kluwer Academic Publishers, Boston, 2003.

\bibitem[AKM05]{achterberg_branching_2005}
Tobias Achterberg, Thorsten Koch, and Alexander Martin.
\newblock Branching rules revisited.
\newblock {\em Operations Research Letters}, 33(1):42--54, January 2005.

\bibitem[BCGL23]{banerjee_graph_2023}
Siddhartha Banerjee, Vincent {Cohen-Addad}, Anupam Gupta, and Zhouzi Li.
\newblock Graph {{Searching}} with {{Predictions}}.
\newblock In {\em 14th {{Innovations}} in {{Theoretical Computer Science
  Conference}} ({{ITCS}} 2023)}, Saarbrücken/Wadern, 2023. Schloss-Dagstuhl.

\bibitem[BDHS23]{baligacs_exploration_2023}
J{\'u}lia Balig{\'a}cs, Yann Disser, Irene Heinrich, and Pascal Schweitzer.
\newblock Exploration of {{Graphs}} with {{Excluded Minors}}.
\newblock In {\em 31st Annual European Symposium on Algorithms (ESA 2023)},
  Saarbrücken/Wadern, 2023. Schloss Dagstuhl.

\bibitem[BDSV18]{balcan_learning_2018}
Maria-Florina Balcan, Travis Dick, T.~Sandholm, and Ellen Vitercik.
\newblock Learning to {{Branch}}.
\newblock {\em ICML}, 2018.

\bibitem[Ber98]{Berman1998}
Piotr Berman.
\newblock {\em On-line searching and navigation}, pages 232--241.
\newblock Springer Berlin Heidelberg, Berlin, Heidelberg, 1998.

\bibitem[CP99]{clausen_best_nodate}
Jens Clausen and Michael Perregaard.
\newblock On the best search strategy in parallel branch-and-bound: Best-first
  search versus lazy depth-first search.
\newblock {\em Annals of Operations Research}, 90:1--17, 1999.

\bibitem[DCD95]{goos_near_1995}
Pallab Dasgupta, P.~P. Chakrabarti, and S.~C. DeSarkar.
\newblock A near optimal algorithm for the extended cow-path problem in the
  presence of relative errors.
\newblock In {\em Foundations of {{Software Technology}} and {{Theoretical
  Computer Science}}}, pages 22--36. {Springer Berlin Heidelberg}, {Berlin,
  Heidelberg}, 1995.

\bibitem[DFKP04]{diks_tree_2004}
Krzysztof Diks, Pierre Fraigniaud, Evangelos Kranakis, and Andrzej Pelc.
\newblock Tree exploration with little memory.
\newblock {\em Journal of Algorithms}, 51(1):38--63, April 2004.

\bibitem[Fre93]{frederickson_optimal_1993}
G.N. Frederickson.
\newblock An {{Optimal Algorithm}} for {{Selection}} in a {{Min-Heap}}.
\newblock {\em Information and Computation}, 104(2):197--214, June 1993.

\bibitem[Gle22]{gleixner_personal}
Ambros~M Gleixner.
\newblock personal communication, November 2022.

\bibitem[KP94]{kalyanasundaram_constructing_1994}
Bala Kalyanasundaram and Kirk~R. Pruhs.
\newblock Constructing competitive tours from local information.
\newblock {\em Theoretical Computer Science}, 130(1):125--138, August 1994.

\bibitem[KSW86]{KSW86}
Richard~M Karp, Michael~E Saks, and Avi Wigderson.
\newblock On a search problem related to branch-and-bound procedures.
\newblock In {\em 27th Annual Symposium on Foundations of Computer Science},
  pages 19--28, Washington, DC, 1986. IEEE Computer Society.

\bibitem[LS99]{linderoth_computational_1999-1}
J.~T. Linderoth and M.~W.~P. Savelsbergh.
\newblock A {{Computational Study}} of {{Search Strategies}} for {{Mixed
  Integer Programming}}.
\newblock {\em INFORMS Journal on Computing}, (2):173--187, May 1999.

\bibitem[LZ17]{lodi_learning_2017-1}
Andrea Lodi and Giulia Zarpellon.
\newblock On learning and branching: A survey.
\newblock {\em TOP}, 25(2):207--236, July 2017.

\bibitem[MJSS16]{morrison_branch-and-bound_2016-1}
David~R. Morrison, Sheldon~H. Jacobson, Jason~J. Sauppe, and Edward~C. Sewell.
\newblock Branch-and-bound algorithms: A survey of recent advances in
  searching, branching, and pruning.
\newblock {\em Discrete Optimization}, 19:79--102, February 2016.

\bibitem[MMS12]{megow_online_2012}
Nicole Megow, Kurt Mehlhorn, and Pascal Schweitzer.
\newblock Online graph exploration: {{New}} results on old and new algorithms.
\newblock {\em Theoretical Computer Science}, 463:62--72, December 2012.

\bibitem[MP80]{MP80}
J.I. Munro and M.S. Paterson.
\newblock Selection and sorting with limited storage.
\newblock {\em Theoretical Computer Science}, 12(3):315--323, 1980.

\bibitem[MU05]{mitzenmacher_probability_2005-2}
Michael Mitzenmacher and Eli Upfal.
\newblock {\em Probability and Computing: An Introduction to Randomized
  Algorithms and Probabilistic Analysis}.
\newblock {Cambridge University Press}, {New York}, 2005.

\bibitem[PPSV15]{pietracaprina_space-ecient_nodate}
Andrea Pietracaprina, Geppino Pucci, Francesco Silvestri, and Fabio Vandin.
\newblock Space-efficient parallel algorithms for combinatorial search
  problems.
\newblock {\em Journal of Parallel and Distributed Computing}, 76:58--65, 2015.

\bibitem[SS93]{suhl_fast_1993}
Leena~M. Suhl and Uwe~H. Suhl.
\newblock A fast {{LU}} update for linear programming.
\newblock {\em Annals of Operations Research}, 43(1):33--47, January 1993.

\bibitem[{Tho}06]{kamphans}
{Thomas Kamphans}.
\newblock {\em Models and Algorithms for Online Exploration and Search}.
\newblock PhD thesis, Rheinische Friedrich-Wilhelms-Universität Bonn, 2006.

\end{thebibliography}

\end{document}